\documentclass[a4paper,11pt]{article}

\usepackage{fullpage}
\usepackage{times}
\usepackage{soul}
\usepackage{url}
\usepackage[utf8]{inputenc}
\usepackage[small]{caption}
\usepackage{graphicx}
\usepackage{amsmath}
\usepackage{booktabs}

\usepackage[ruled,vlined]{algorithm2e}
\SetArgSty{textrm}

\usepackage{enumitem}

\usepackage{libertine}

\usepackage{amssymb,amsmath,amsfonts,amstext,amsthm}

\usepackage[breaklinks=true]{hyperref}
\usepackage[svgnames]{xcolor}
\usepackage[capitalise,nameinlink]{cleveref}
\hypersetup{colorlinks={true},linkcolor={DarkBlue},citecolor=[named]{DarkGreen}}

\usepackage{tikz}  
\usetikzlibrary{arrows}
\usetikzlibrary{patterns,snakes}
\usetikzlibrary{decorations.shapes}
\tikzstyle{overbrace text style}=[font=\tiny, above, pos=.5, yshift=5pt]
\tikzstyle{overbrace style}=[decorate,decoration={brace,raise=5pt,amplitude=3pt}]
\usetikzlibrary{shapes.geometric}
\usetikzlibrary{shapes,positioning}
\usetikzlibrary{calc,positioning}
\usetikzlibrary{shapes.multipart}

\definecolor{cadmiumgreen}{rgb}{0.0, 0.42, 0.24}

\usepackage{bbm}

\newtheorem{theorem}{Theorem}[section]

\newtheorem{lemma}[theorem]{Lemma}

\theoremstyle{definition}

\usepackage{natbib}

\usepackage{mathtools}
\usepackage{xcolor} 

\usepackage{authblk}

\usepackage{subcaption}

\usepackage{mdframed}
\mdfsetup{linewidth=1pt}

\usepackage{xspace}
\usepackage{fancyhdr}
\usepackage{tcolorbox}

\newcommand{\one}[1]{\mathbf{1}\left\{#1\right\}}

\newcommand{\cost}{\text{cost}}
\newcommand{\bz}{\mathbf{z}}
\newcommand{\bt}{\mathbf{t}}
\newcommand{\bx}{\mathbf{x}}

\newcommand{\SC}{\text{\normalfont SC}}
\newcommand{\MC}{\text{\normalfont MC}}

\begin{document}

\allowdisplaybreaks

\title{\bf On Discrete Truthful Heterogeneous \\Two-Facility Location}

\author{Panagiotis Kanellopoulos}
\author{Alexandros A. Voudouris}
\author{Rongsen Zhang}

\affil{School of Computer Science and Electronic Engineering \\ University of Essex, UK}

\renewcommand\Authands{ and }
\date{}

\maketitle

\begin{abstract}
We revisit the discrete heterogeneous two-facility location problem, in which there is a set of agents that occupy nodes of a line graph, and have private approval preferences over two facilities. When the facilities are located at some nodes of the line, each agent derives a cost that is equal to her total distance from the facilities she approves. The goal is to decide where to locate the two facilities, so as to (a) incentivize the agents to truthfully report their preferences, and (b) achieve a good approximation of the minimum total (social) cost or the maximum cost among all agents. For both objectives, we design deterministic strategyproof mechanisms with approximation ratios that significantly outperform the state-of-the-art, and complement these results with (almost) tight lower bounds. 
\end{abstract}

\section{Introduction}
In the classic truthful single-facility location problem, there is a set of agents with private positions on the line of real numbers, and a public facility (such as a library or a park) whose location we need to decide. This decision must be made so as to (a) incentivize the agents to truthfully reveal their positions (an agent would be willing to lie if this leads to the facility being located closer to her true position), and (b) optimize a social objective (such as the total distance of the agents from the facility location, or the maximum distance). Since the celebrated paper of \citet{procaccia09approximate}, this problem and its many variants have been studied extensively in the literature on {\em approximate mechanism design without money}. For a comprehensive introduction to the various different facility location models that have been considered over the years, we refer the interested reader to the recent survey of~\citet{survey21}.

A recent stream of papers have focused on {\em heterogeneous facility location} problems, with multiple facilities (typically, two) that are different in nature (e.g., a school and a bar). As such, the agents care both for the location and the types of the facilities, aiming for the facilities they like the most to be as close to their position as possible. To give an example, a family would like to be closer to a school than to a bar, whereas a single person might want the opposite. Many settings have been proposed to model the different preferences the agents may have about the facilities (see the discussion in the related work). With few exceptions, all of these models assume that both facilities can be placed at any point of the real line, even at the same one. \citet{serafino2016} deviated from these assumptions and studied a discrete version of the problem, where the line is a discrete graph, the agents occupy nodes of the line (which is common knowledge) and have {\em approval} preferences over two facilities, which can only be placed at {\em different} nodes of the line. Given facility locations, the cost of an agent is defined as the total distance from the facilities she approves.

\citeauthor{serafino2016} presented bounds on the approximation ratio of deterministic and randomized mechanisms in terms of both social objectives of interest. In particular, for the social cost, they showed that the best possible approximation ratio of deterministic mechanisms is between $9/8$ and $n-1$, where $n$ is the number of agents. In contrast, they designed a randomized mechanism that always outputs a solution with minimum expected social cost. For the maximum cost objective, they showed that the best possible approximation ratio of deterministic mechanisms is between $3/2$ and $3$, and that of randomized mechanisms is between $4/3$ and $3/2$. In this paper, we focus exclusively on deterministic mechanisms, and improve upon the bounds of \citeauthor{serafino2016} for both the social and the maximum cost. 

\subsection{Our contribution}

The upper bounds shown by \citeauthor{serafino2016} for the social and the maximum cost both follow by the same deterministic mechanism, named {\sc TwoExtremes}, which works along the lines of the mechanism considered by \citet{procaccia09approximate} for homogeneous $2$-facility location. {\sc TwoExtremes} locates one of the facilities at the node occupied by the leftmost agent among those that approve it, and the other facility at the node occupied by the rightmost agent among those that approve it; in case of a collision, one of the facilities is moved a node to the left or the right. There are two main reasons for the deficiency of {\sc TwoExtremes}: (i) the boundary agents (leftmost and rightmost) among those that approve a facility may be rather far away from the median such agent, whose node would be the ideal location for the facility, and (ii), it does not exploit the available information about the position of the agents in any way. Our improved mechanisms take care of these two reasons: We place the facilities closer to median agents (without breaking strategyproofness), and exploit the information about the agent positions of the agents. 

For the social cost, we design the {\sc Fixed-or-Median-Nearest-Empty} (FMNE) mechanism with an approximation ratio of at most $17/4=4.25$. The mechanism switches between two cases based on the structure of the line: If there are no empty nodes, it fixes the locations of the facilities to be the two central nodes of the line; otherwise, if there are empty nodes, it locates one of the facilities at the position of the median agent among those that approve it, and the other facility at one of the nearest empty nodes to the median agent among those that approve it. We complement this result with an improved lower bound of $4/3$ on the approximation ratio of all mechanisms, which follows by two instances with only three agents and no empty nodes. Motivated by this lower bound construction, we then focus on instances with three agents, and design the $3$-agent {\sc Priority-Dictatorship} mechanism that achieves the best-possible bound of $4/3$.

For the maximum cost, we design a parameterized class of mechanisms {\sc $\alpha$-Left-Right} ($\alpha$-LR), each member of which partitions the line into two parts, from the first node to node $\alpha$, and from node $\alpha+1$ to the last node. Then, the decision about the locations of the facilities is based on the preferences of the agents included in the two parts. We show that all mechanisms of the class are strategyproof, and there are members with approximation ratio at most $2$. In particular, when the size $m$ of the line is an even number, the bound is achieved by $m/2$-LR, and when $m$ is odd, it is achieved by $(m+1)/2$-LR. Finally, we show a tight lower bound of $2$ on the approximation ratio of all strategyproof mechanisms, using a construction involving a sequence of five instances with three agents and no empty nodes. 

\begin{table}[t]
    \centering
    \begin{tabular}{c|cc}
                        & Lower bound & Upper bound \\ \hline 
       Social cost      &   $4/3^\star$ ($9/8$)    & $17/4$ ($n-1$)      \\
       Maximum cost     &   $2$ ($3/2$)      & $2$ ($3$)         \\ \hline
    \end{tabular}
    \caption{An overview of our bounds on the approximation ratio of deterministic strategyproof mechanisms for the social cost and the maximum cost. The bounds in parentheses are the previously best known ones shown by \citet{serafino2016}. The lower bound of $4/3$ marked with a $\star$ is tight for instances with three agents.}
    \label{tab:results}
\end{table}

An overview of our bounds, and how they compare to the previously best ones shown by \citet{serafino2016}, is given in Table~\ref{tab:results}.

\subsection{Related work}
As already mentioned above, the survey of \citet{survey21} nicely discusses the many different facility models that have been considered over the years in the literature on approximate mechanism design without money. Here, we will mainly discuss papers on heterogeneous facility location models that are closely related to ours. Besides the work of \citet{serafino2016}, our paper here, a few ones on characterizations of onto strategyproof mechanisms for the single-facility location problem in discrete lines, cycles~\citep{DFMN} and trees \citep{FMfacility}, and some related to truthful single-winner voting on a line metric (e.g., \citep{feldman2016voting,FV2021facility}), the related literature has mainly focused on continuous settings, where the facilities can be located at any point of the line, even the same one. 

The first heterogeneous facility location model, combining elements from the classic single-facility location problem and the obnoxious single-facility location problem~\citep{cheng2011obnoxious,cheng2013obnoxious}, was independently proposed and studied by \citet{feigenbaum2015hybrid} and \citet{zou2015dual}. In this setting, there are two facilities to be located on the real line, and the agents have {\em dual} preferences over the facilities; that is, an agent likes or dislikes a facility. The authors showed bounds on the approximation ratio of deterministic and randomized strategyproof mechanisms for different cases depending on whether the positions or the preferences of the agents are their private information (and can thus lie about them). \citet{kyropoulou2019constrained} considered an extension of this model, where the location space of the two facilities is a constrained region of the Euclidean space. 

\citet{chen2020optional} studied a setting with agents that have {\em optional} (or, {\em approval}) preferences over the facilities; that is, an agent either likes a facility or is indifferent about it. The authors consider two different cost functions of the agents, one that is equal to the distance from the closest facility that the agent approves, and one that is equal to the distance from the farthest such facility. \citet{li2020optional-improved} studied an extension of this setting in more general metrics (beyond the line), and designed mechanisms that improve some of results of \citeauthor{chen2020optional}. \citet{deligkas2021limited} considered a similar preference model, but with the difference that the goal is to locate just one of the facilities (and, more generally, $k$ out of $m$), and not all of them. 

\citet{anastasiadis2018heterogeneous} considered a model that combines dual and optional preferences, in the sense that the agents can like, dislike or be indifferent about a facility. \citet{fong2018fractional} studied  a setting with {\em fractional} preferences, where each agent has a weight in $[0,1]$ for each facility indicating how much she likes it. Finally, \citet{Xu2021minimumdistance} focused on the problem where the locations of the locations of the two facilities must satisfy a minimum distance requirement.

\section{Preliminaries}
We consider the discrete two-facility location problem. An instance $I$ of this problem consists of a set $N$ of $n \geq 2$ {\em agents}, two {\em facilities}, and a line graph with $m \geq n$ nodes. Each agent occupies a node $x_i$ of the line, such that different agents occupy different nodes. By $\bx$ we denote the {\em position profile} consisting of the positions of all agents (i.e., the nodes they occupy) as well as the positions of possible empty nodes; the position profile is assumed to be {\em common knowledge}. Every agent $i$ also has a private {\em approval preference} $\bt_{i} \in \{0,1\}^2$ over the two facilities: If $t_{ij}=1$, agent $i \in N$ approves facility $j \in \{1,2\}$; otherwise, she disapproves it. By $\bt = (\bt_i)_{i \in N}$ we denote the {\em preference profile} consisting of the preferences of all agents. It will be useful to denote by $N_j$ the set of agents that approve facility $j \in \{1,2\}$. Clearly, the two sets $N_1$ and $N_2$ need not be disjoint if there are agents that approve both facilities. As $\bx$ and $\bt$ implicitly include all the information related to an instance, we denote $I=(\bx,\bt)$.

A {\em feasible solution} $\bz=(z_1,z_2)$ determines the node $z_j$ where each facility $j \in \{1,2\}$ is located, so that $z_1 \neq z_2$.
Given a feasible solution $\bz$, the cost of any agent $i$ in instance $I$ is her {\em total} distance from the facilities she approves, i.e., 
\begin{align*}
\cost_i(\bz|I) = \sum_{j \in \{1,2\}} t_{ij} \cdot d(i,j),
\end{align*}
where $d(i,j) = |x_i-z_j|$ is the distance between agent $i$ and facility $j$. 

A {\em mechanism} takes as input an instance and outputs a feasible solution. A mechanism $M$ is said to be {\em strategyproof} if the solution $M(I)$ it computes when given as input the instance $I=(\bx,\bt)$ is such that no agent $i$ has incentive to report a false preference $\bt_i' \neq \bt_i$ to decrease her cost, i.e., 
\begin{align*}
\cost_i(M(I)|I) \leq \cost_i(M(\bx, (\bt_i',\bt_{-i}))|I),
\end{align*}
where $(\bt_i',\bt_{-i})$ is the preference profile according to which agent $i$'s preference is $\bt_i'$, while the preference of any other agent is the same as in $\bt$.

We consider two well-known social objectives, which are functions of feasible solutions. 
Given an instance $I$, the {\em social cost} of a feasible solution $\bz$ is the total cost of all the agents, i.e., 
\begin{align*}
\SC(\bz|I) = \sum_{i \in N} \cost_i(\bz|I).
\end{align*}
The {\em max cost} of $\bz$ is the maximum cost among all agents, i.e., 
\begin{align*}
\MC(\bz|I) = \max_{i \in N} \cost_i(\bz|I).
\end{align*}
Let $\SC^*(I)=\min_z \SC(z|I)$ be the minimum possible social cost for instance $I$, achieved by any feasible solution. Similarly, let $\MC^*(I)=\min_z \MC(z|I)$ be the minimum possible maximum cost for $I$.

For any social objective $f \in \{\SC,\MC\}$, the {\em $f$-approximation ratio} of a mechanism $M$ is the worst-case ratio (over all possible instances)  between the objective value of the solution computed by $M$ over the minimum possible objective value among all feasible solutions, i.e.,
\begin{align*}
\rho(M) = \sup_{I} \frac{f(M(I)|I)}{f^*(I)}.
\end{align*}
Our goal is to design strategyproof mechanisms with an as low $f$-approximation ratio as possible (close to $1$). 


\section{A constant-approximation strategyproof mechanism for the social cost}
We start with the social cost objective. For general instances with $n$ agents, we design the strategyproof mechanism {\sc Fixed-or-Median-Nearest-Empty} (FMNE) with approximation ratio $17/4$, thus greatly improving upon the previous bound of $n-1$ of \citet{serafino2016}. Our mechanism exploits the known information about the position profile, and distinguishes between two cases depending on whether the given instance contains empty nodes or not. If there are no empty nodes, FMNE locates the facilities next to each other at central nodes of the line (in particular, nodes $\lfloor n/2 \rfloor$ and $\lfloor n/2 \rfloor+1$); this is the {\sc Fixed} part of the mechanism. If there are empty nodes, FMNE locates facility $1$ at the node occupied by the median agent among those that approve facility $1$, and facility $2$ at the empty node that is nearest to the node occupied by the median agent among those that approve facility $2$; this is the {\sc Median-Nearest-Empty} part of the mechanism. See Mechanism~\ref{mech:FMNE}. 

\newcommand\mycommfont[1]{\normalfont\textcolor{blue}{#1}}
\SetCommentSty{mycommfont}
\begin{algorithm}[t]
\SetNoFillComment
\caption{\sc Fixed-or-Median-Nearest-Empty (FMNE)}
\label{mech:FMNE}
{\bf Input:} Instance $I$ with $n$ agents \;
{\bf Output:} Feasible solution $\bz = (z_1,z_2)$ \;
\uIf(\tcp*[h]{{\sc Fixed} part}){there are no empty nodes}{ 
    $z_1 \gets \lfloor n/2 \rfloor$\;
    $z_2 \gets \lfloor n/2 \rfloor +1$\;
}
\Else(\tcp*[h]{{\sc Median-Nearest-Empty} part}){
    \For{$j \in \{1,2\}$}{
        $y_j \gets $ position of the leftmost median agent in $N_j$\;
    }
    $z_1 \gets y_1$\;
    $z_2 \gets $ nearest empty node to $y_2$, breaking ties in favor of the rightmost such empty node\;
}
\end{algorithm}

\begin{theorem} \label{thm:FMNE-sp}
FMNE is strategyproof.
\end{theorem}

\begin{proof}
Consider an arbitrary instance $I$. The mechanism is clearly strategyproof if there are no empty nodes in $I$ as the locations of the facilities are fixed and independent of the preferences of the agents. So, it remains to consider the case where $I$ contains empty nodes. Let $i$ be an arbitrary agent. We switch between the following three cases:

\medskip
\noindent
{\bf Agent $i$ approves only facility $1$ ($i \in N_1 \setminus N_2$).} 
Suppose without loss of generality that $x_i \leq y_1$. Any misreport of agent $i$ can only lead to a median $y_1'$ among the agents that approve facility $1$ which is farther away from $x_i$. In particular, if $i$ misreports that she approves only facility $2$, then $y_1' \geq y_1$, whereas if $i$ misreports that she approves both facilities, then $y_1' = y_1$.

\medskip
\noindent
{\bf Agent $i$ approves only facility $2$ ($i \in N_2 \setminus N_1$).} 
Suppose without loss of generality that facility $2$ is positioned at some empty node $e$ with $x_e > y_2$. Denote by $y_2'$ the median node occupied by the agents that approve facility $2$ when $i$ misreports. 
\begin{itemize}
\item 
If agent $i$ misreports that she approves both facilities, then $y_2' = y_2$, and hence the position of facility $2$ as well as the cost of agent $i$ remain the same.

\item
If $x_i \leq y_2$ and agent $i$ misreports that she only approves facility $1$, then $y_2' \geq y_2$. As a result, either $e$ continues to be the nearest empty node to $y_2'$ and the cost of $i$ remains exactly the same, or another empty node $e'$ with $x_{e'} > x_e$ becomes the nearest empty node to $y_2'$, and the cost of $i$ strictly increases. 

\item 
If $x_i > y_2$ and agent $i$ misreports that she only approves facility $1$, then $y_2' \leq y_2$. As a result, either $e$ continues to be the nearest empty node to $y_2'$, or another empty node $e$ with $x_{e'} < y_2 < x_e$ becomes the nearest empty node to $y_2'$. In any case, the cost of $i$ does not decrease.
\end{itemize}

\noindent
{\bf Agent $i$ approves both facilities ($i \in N_1 \cap N_2$).}
Since the cost of agent $i$ is the sum of costs she derives from the two facilities and we decide where to locate each facility independently from the other facility, the same arguments for the previous two cases show that no possible misreport can lead to a strictly lower cost. 
\end{proof}

To argue about the approximation ratio of FMNE, we will distinguish between instances with and without empty nodes. In our proofs, we exploit the following lower bounds on the optimal social cost; we include the proof for completeness.
Here, $\one{X}$ is equal to $1$ if the event $X$ is true, and $0$ otherwise.

\begin{lemma}[\citep{serafino2016}]\label{lem:opt-sc-bound}
For any instance $I$ in which there are $n_j$ agents that approve facility $j \in \{1,2\}$, it holds that
\begin{align*}
\SC^*(I) &\geq \frac{1}{4}\bigg(n_1^2 + n_2^2 -\one{n_1 \text{ odd}}-\one{n_2 \text{ odd}}\bigg) \\
&\geq \frac{1}{4}\bigg(n_1^2 + n_2^2 - 2\bigg).
\end{align*}
\end{lemma}
\begin{proof}
We argue about each facility $j\in\{1,2\}$ independently. When $n_j$ is odd, the optimal allocation can, at best, have facility $j$ at one of these $n_j$ nodes and have two agents at distance $i$, for $i\in\{1, \dots, \frac{n_j-1}{2}\}$ from $j$; the total cost due to facility $j$ is then $\frac{n_j^2-1}{4}$. When $n_j$ is even, the optimal allocation can, at best, place facility $j$ facility at one of the $n_j$ nodes and have two agents at distance $i$, for $i\in\{1, \dots, \frac{n_j}{2}-1\}$, and an agent at distance $\frac{n_j}{2}$; the total cost due to facility $j$ in this case is then~$\frac{n_j^2}{4}$.
\end{proof}

For instances without empty nodes and $n \geq 5$, we will show that the approximation ratio of FMNE (in particular, its {\sc Fixed} part) is at most $3$; note that for $n \leq 4$, the {\sc TwoExtremes} mechanism of \citet{serafino2016} is $3$-approximate.

\begin{theorem} \label{thm:FMNE-Fixed-approx}
For any instance with $n \geq 5$ agents and no empty nodes, the $\SC$-approximation ratio of FMNE is at most $3$.
\end{theorem}

\begin{proof}
Consider an instance $I$ and recall that $N_j$ denotes the set of agents that approve facility $j$. Let $n_1=|N_1|$, $n_2=|N_2|$, and $b=|N_1 \cap N_2|$; clearly, it holds that $n=n_1+n_2-b$.

We first consider the case where $n$ is even, i.e., $n\geq 6$. For any agent $i$, with $i\in\{1, \dots, n\}$, the maximum distance of $i$ to a facility is $|n/2+1-i|$. Furthermore, each of the $b$ agents that approve both facilities faces an added cost of at most $n/2$ due to the distance to the agent's nearest facility. Therefore, the total cost of the solution $\bz$ computed by the mechanism is bounded by 
\begin{align*}
\SC(\bz|I)&\leq 2\sum^{n/2}_{i=1}{i}+b\cdot\frac{n}{2}\\
&= \frac{n}{2} \cdot \left(\frac{n}{2}+1\right) + b\cdot\frac{n}{2}\\
&= \frac{n_1^2+n_2^2+b^2+2n_1n_2-2bn_1-2bn_2+2n_1+2n_2-2b}{4}+\frac{bn_1+bn_2-b^2}{2}\\
&\leq \frac{n_1^2+n_2^2+2n_1n_2+2n_1+2n_2}{4},
\end{align*}
where the second equality holds since $n=n_1+n_2-b$, while the last inequality holds since $b\geq 0$.

By Lemma~\ref{lem:opt-sc-bound}, $\SC^*(I) \geq \frac{1}{4}\left(n_1^2+n_2^2-2\right)$, and thus the approximation ratio is bounded by 
\begin{align*}
\frac{\SC(\bz|I)}{\SC^*(I)}&\leq \frac{n_1^2+n_2^2+2n_1n_2+2n_1+2n_2}{n_1^2+n_2^2-2}.
\end{align*}
To prove the claim, it suffices to show that, when $n_1+n_2\geq 6$, it holds that $n_1^2+n_2^2+2n_1n_2+2n_1+2n_2\leq 3n_1^2+3n_2^2-6$, i.e., $(n_1-n_2)^2+n_1^2+n_2^2\geq 2n_1+2n_2+6$. Observe that, when $n_1+n_2\geq 6$, it holds that  $n_1^2+n_2^2\geq 3(n_1+n_2)\geq 2n_1+2n_2+6$; the claim follows.

We now consider the case where $n\geq 5$ is odd; the analysis is slightly more involved, but follows along similar lines. 
Observe that the maximum distance of any agent $i$ positioned at some of the first $(n-1)/2$ nodes from a facility (in particular, facility $2$) is $(n+1)/2-i$, while the maximum distance of any agent $i$ positioned at some of the last $(n+1)/2$ nodes from a facility (in this case, facility $1$) is $i-(n-1)/2$. Furthermore, each of the $b$ agents that approve both facilities faces an added cost of at most $(n-1)/2$. So, the total cost of the solution $\bz$ computed by the mechanism is bounded by 
\begin{align*}
\SC(\bz|I)&\leq \sum^{(n-1)/2}_{i=1}{i}+\sum^{(n+1)/2}_{i=1}{i}+b\cdot\frac{n-1}{2}\\
&= \frac{(n+1)^2}{4}+b\cdot\frac{n-1}{2}\\
&= \frac{n_1^2+n_2^2+b^2+2n_1n_2-2bn_1-2bn_2+2n_1+2n_2-2b+1}{4}+\frac{bn_1+bn_2-b^2-b}{2}\\
&= \frac{n_1^2+n_2^2+2n_1n_2+2n_1+2n_2+1-b^2-4b}{4},
\end{align*}
where, again, the second equality holds since $n=n_1+n_2-b$.

By Lemma~\ref{lem:opt-sc-bound}, $\SC^*(I) \geq \frac{1}{4}\left( n_1^2+n_2^2-\one{n_1 \text{ odd}}-\one{n_2 \text{ odd}}\right)$, and thus the approximation ratio is bounded by 
\begin{align*}
\frac{\SC(\bz|I)}{\SC^*(I)}&\leq \frac{n_1^2+n_2^2+2n_1n_2+2n_1+2n_2+1-b^2-4b}{n_1^2+n_2^2-\one{n_1 \text{ odd}}-\one{n_2 \text{ odd}}}.
\end{align*}
To prove the claim it suffices to show that $(n_1-n_2)^2+n_1^2+n_2^2+b^2+4b\geq 2n_1+2n_1+1+3(\one{n_1 \text{ odd}}+\one{n_2 \text{ odd}})$. If $b\geq 1$, then $n_1+n_2\geq 6$, and the claim follows since $(n_1-n_2)^2+n_1^2+n_2^2 \geq 2(n_1+n_2+1)$ holds in this case. Otherwise, when $b=0$, then exactly one of $n_1, n_2$ is odd and it suffices to show that $(n_1-n_2)^2+n_1^2+n_2^2\geq 2n_1+2n_1+4$.
Since $(n_1-n_2)^2\geq 1$, this always holds if $n_1+n_2\geq 5$.
\end{proof}

For instances with at least one empty node, we will show that the approximation ratio of FMNE (in particular, its {\sc Median-Nearest-Empty} part) is $17/4$ for any $n \geq 6$; observe that the {\sc TwoExtremes} mechanism of \citet{serafino2016} achieves an approximation ratio of at most $4$ when $n \leq 5$. 
Our proof for the approximation ratio of FMNE in this case relies on the following technical lemma.

\begin{lemma}\label{lem:technical}
Let $f(x,y) = \frac{y^2+4xy+2y+1}{x^2+y^2-2}$. For non-negative integers $x,y$ such that $x+y\geq 6$, it holds $f(x,y)\leq 13/4$. 
\end{lemma}

\begin{proof}
First, observe that $f(x,y)$ can be written as $f(x,y) = 1+\frac{-x^2+4xy+2y+3}{x^2+y^2-2}$. It suffices to limit our attention to the values of $x,y$ for which $f(x,y)>3$, i.e., to these $x,y$ such that $\frac{-x^2+4xy+2y+3}{x^2+y^2-2} > 2$. By rearranging, we obtain $-x^2+4xy+2y+3 > 2x^2+2y^2-4$, and, therefore, $-x^2+2y+7 > 2(x-y)^2$.

Let $y=x+k$ for some integer $k$ and rewrite the last inequality as $-x^2+2x+7> 2(k^2-k)$. Clearly, for $x\geq 4$ the inequality never holds as the left-hand-side term is negative and the right-hand-side term is always non-negative. Hence, we obtain that $f(x\geq 4, y)< 13/4$. For the remaining values of $x$, i.e., when $x\in \{0,1,2,3\}$, recall that $x+y\geq 6$, i.e., $2x+k\geq 6$. When $x=0$, it must be $k\geq 6$ and the inequality does not hold, as $7<60$. When $x=1$, we have $k\geq 4$ and, again, the inequality does not hold, as $8< 24$. For $x=2$, we obtain $k\geq 2$ and the inequality becomes $7>2(k^2-k)$, which holds only when $k=2$; in this case, $f(2,4) = 19/6<13/4$. Finally, for $x=3$, we have that $k\geq 0$ and the inequality becomes $4 > 2(k^2-k)$ which holds for $k\in \{0,1\}$. The proof follows by observing that $\max\{f(3,3), f(3,4)\} = \max\{13/4, 73/23\}\leq 13/4.$
\end{proof}

We are now ready to prove the bound for instances with empty nodes.

\begin{theorem}
For any instance with $n \geq 6$ and at least one empty node, the $\SC$-approximation ratio of FMNE is at most $17/4$. 
\end{theorem}

\begin{proof}
Consider any instance $I$. We first argue a bit about the optimal social cost of $I$. A solution that minimizes the social cost locates each facility $j \in \{1,2\}$ to the node $y_j$ occupied by a median agent in $N_j$. However, this solution might not be feasible if $y_1=y_2$, and so the optimal social cost can only be larger. We have that
\begin{align}\label{eq:opt-bound}
\SC^*(I) \geq \sum_{i \in N_1} d(x_i, y_1) + \sum_{i \in N_2} d(x_i, y_2). 
\end{align}

Now, let us focus on the social cost of the solution $\bz$ computed by the mechanism. Let $e$ be the empty node where facility $2$ is located; without loss of generality, we can assume that $x_e > y_2$. Combined with the fact that facility $1$ is located at $y_1$, we have that $\bz = (y_1,x_e)$, and 
\begin{align*}
\SC(\bz | I) = \sum_{i \in N_1} d(x_i,y_1) + \sum_{i \in N_2} d(x_i,x_e).
\end{align*}
The first term appears in the lower bound of the optimal social cost given by Inequality \eqref{eq:opt-bound}, so all we need to do is bound the second term of the above expression. 

We partition the set $N_2$ into three sets $L$, $M$ and $R$ depending on the positions of the agents in $N_2$ compared to $y_2$ and $x_e$, as follows:
\begin{itemize}
\item $L = \{i \in N_2: x_i \leq y_2\}$;
\item $M = \{i \in N_2: x_i \in (y_2, x_e) \}$;
\item $R = \{i \in N_2: x_i > x_e\}$.
\end{itemize}
By the definition of median, we have that $|L| \geq |M|+|R|$; in particular, this is an equality if $n_2 = |N_2|$ is even, and a strict inequality if $n_2$ is odd (as $L$ also includes the median agent in this case). 
Observe that
\begin{itemize}
\item For every agent $i \in M$, there exists a unique agent in $j \in L$ such that 
$$d(x_j,x_e) = d(x_j, x_i) + d(x_i,x_e) = d(x_j,y_2) + d(x_i,y_2) + d(x_i,x_e).$$
\item For every agent $i \in R$, there exists a unique agent $j \in L$ such that
$$d(x_j,x_e) + d(x_i,x_e) = d(x_j,y_2) + d(y_2,x_e) + d(x_i,x_e) = d(x_j,y_2) + d(x_i,y_2).$$
\end{itemize}
Hence, we have that
\begin{align*}
\sum_{i \in N_2} d(x_i,x_e) 
&= \sum_{i \in L} d(x_i,x_e)  + \sum_{i \in M} d(x_i,x_e)  + \sum_{i \in R} d(x_i,x_e) \\
&\leq \sum_{i \in N_2} d(x_i,y_2) + d(y_2,x_e)\one{n_2 \text{ odd}} + 2\cdot \sum_{i \in M} d(x_i,x_e).
\end{align*}

Next, we will bound the second and third terms of the above expression. Since each agent occupies a different node, we can upper-bound the total distance of the agents in $M$ as follows:
\begin{align*}
&d(y_2,x_e)\one{n_2 \text{ odd}} + 2\cdot \sum_{i \in M} d(x_i,x_e) \\
&\leq d(y_2,x_e)\one{n_2 \text{ odd}}
+ 2\cdot \bigg( d(y_2,x_e)-1 + d(y_2,x_e)-2 + \ldots + d(y_2,x_e)-|M| \bigg) \\
&= - |M|^2 + (2d(y_2,x_e)-1) |M| + d(y_2,x_e)\one{n_2 \text{ odd}}.
\end{align*}
Now observe that $d(y_2,x_e) > |M|$ (since all agents in $M$ are between $y_2$ and $e$); thus, the last expression in the above derivation is an increasing function in terms of $|M|$. It is clearly also an increasing function in terms of $d(y_2,x_e)$. Since 
$|M| \leq \frac{1}{2}(n_2-\one{n_2 \text{ odd}})$ and 
$d(y_2,x_e) \leq n_1 + 1+|M| \leq n_1 + 1+ \frac{1}{2}(n_2-\one{n_2 \text{ odd}})$, 
by doing calculations and also using the fact that $\one{n_2 \text{ odd}} \leq 1$,
we obtain
\begin{align*}
d(y_2,x_e)\one{n_2 \text{ odd}} + 2\cdot \sum_{i \in M} d(x_i,x_e) 
&\leq \frac14 \bigg( n_2^2 + 4n_1n_2 + 2n_2+1 \bigg).
\end{align*}
By putting everything together, we have
\begin{align*}
\SC(\bz | I) 
&\leq \sum_{i \in N_1} d(x_i, y_1) + \sum_{i \in N_2} d(x_i, y_2) + \frac14 \bigg( n_2^2 + 4n_1n_2 + 2n_2+1 \bigg) \\
&\leq \SC^*(I) + \frac14 \bigg( n_2^2 + 4n_1n_2 + 2n_2+1 \bigg).
\end{align*}
By Lemma~\ref{lem:opt-sc-bound}, we have $\SC^*(I)\geq \frac{1}{4}\left(n_1^2 + n_2^2 -2 \right)$, and thus the approximation ratio is bounded by
\begin{align*}
\frac{\SC(\bz | I) }{\SC^*(I)} \leq 1 + \frac{n_2^2 + 4n_1n_2 + 2n_2+1}{n_1^2 + n_2^2 -2}.
\end{align*}
The bound of $17/4$ follows by applying Lemma \ref{lem:technical} with $x=n_1$ and $y=n_2$.
\end{proof}

We conclude this section by showing that our analysis of the approximation ratio of FMNE is tight.

\begin{lemma}
There exists an instance with $n \geq 5$ and no empty nodes such that the $\SC$-approximation ratio of FMNE is at least $3$, and an instance with $n \geq 6$ and at least one empty node such that the $\SC$-approximation ratio of FMNE is at least $17/4$.
\end{lemma}

\begin{proof}
For the {\sc Fixed} part of FMNE consider the following instance $I_1$ with $5$ agents and no empty nodes. The first two agents approve only facility $2$, and the last three agents approve only facility $1$. The mechanism outputs the solution $(2,3)$, that is, it locates facility $1$ at the second node and facility $2$ at the third node. The social cost of this solution is $\SC((2,3)|I_1) = 9$. However, an optimal solution is $(4,2)$ with social cost $\SC^*(I_1) = 3$, leading to an approximation ratio of $3$.

For the {\sc Median-Nearest-Empty} part of FMNE consider the following instance $I_2$ with $6$ agents and one empty node. The first three nodes are occupied by agents that approve only facility $2$, the next three nodes are occupied by agents that approve only facility $1$, and the last node is empty. The mechanism outputs the solution $(5,7)$, that is it locates facility $1$ at node $5$ and facility $2$ at the empty node. This solution has social cost $\SC((5,7)|I_2) = 17$. In contrast, an optimal solution is $(5,2)$ with $\SC^*(I_2)=4$, leading to an approximation ratio of $17/4$.
\end{proof}


\section{A tight bound for instances with three agents}
In this section, we restrict to instances with three agents (and possibly many empty nodes). We show a tight bound of $4/3$ on the approximation ratio of strategyproof mechanisms. In particular, we present a rather simple instance without empty nodes showing that the approximation ratio of any strategyproof mechanism is at least $4/3$; this improves upon the previous lower bound of $9/8$ shown by \citet{serafino2016}. We complement this result by designing a mechanism that achieves the bound of $4/3$ when given as input any instance with three agents.

\begin{theorem} \label{thm:sc-lower}
The $\SC$-approximation ratio of any strategyproof mechanism is at least $4/3$.
\end{theorem}

\begin{proof}
We consider two instances with three agents and no empty nodes. In the first instance $I_1$, all agents approve both facilities. Clearly, any mechanism must locate a facility to the first or the third node (or, perhaps, both). Without loss of generality, suppose the mechanism locates facility $2$ at the third node. 

In the second instance $I_2$, the first two agents approve both facilities, while the third agent approves only facility $2$ (that is, the only difference between $I_1$ and $I_2$ is the preference of the third agent). Since facility $2$ is located at the third node in $I_1$, the same must happen in $I_2$; otherwise, agent $3$ would have cost at least $1$ in $I_2$ and incentive to misreport that she approves both facilities, thus changing $I_2$ to $I_1$, and decreasing her cost to $0$. However, both possible feasible solutions $z_1 = (1,3)$ and $z_2 = (2,3)$ have social cost $4$ in $I_2$, whereas an optimal solution (such as $z^* = (1,2)$) has social cost $3$; the theorem follows.
\end{proof}

Next, we design the $3$-agent mechanism {\sc Priority-Dictatorship}, which is strategyproof and has an approximation ratio of at most $4/3$. Consider any instance with three agents; for convenience, we call the agents $\ell$, $c$, and $r$ and let $x_\ell<x_c<x_r$.  Without loss of generality, we assume that $c$ is closer to $r$ than to $\ell$, that is, $x_r - x_c \leq x_c - x_\ell$. Our mechanism gives priority to the central agent over the right agent, and does not take into account the preference of the left agent at all. In particular, the mechanism locates at $x_c$ one of the facilities that agent $c$ approves, and decides the location of the other facility based on the preference of agent $r$. See Mechanism~\ref{mech:priority-dictatorship} for a formal description. We first show that the mechanism is strategyproof.

\begin{algorithm}[t]
\caption{\sc Priority-Dictatorship}
\label{mech:priority-dictatorship}
{\bf Input:} Instance $I$ with three agents $\ell$, $c$, and $r$ such that $x_\ell < x_c < x_r$ \;
{\bf Output:} Feasible solution $\bz$ \;
\uIf {$c \in N_1 \setminus N_2$}{
    \uIf {$r \in N_2$}{
        $\bz \gets (x_c,x_r)$\;
    }
    \Else{
        $\bz \gets (x_c,x_\ell)$\;
    }
}
\uElseIf{$c \in N_2 \setminus N_1$}{
    \uIf {$r \in N_1$}{
        $\bz \gets (x_r,x_c)$\;
    }
    \Else{
        $\bz \gets (x_\ell,x_c)$\;
    }
}
\Else{
    \uIf{$r \in N_2$}{
        $\bz \gets (x_c,x_c+1)$\;
    }
    \Else{
        $\bz \gets (x_c+1,x_c)$\;
    }
}
\end{algorithm}

\begin{theorem}
{\sc Priority-Dictatorship} is strategyproof.
\end{theorem}

\begin{proof}
Consider any instance with three agents $\ell$, $c$ and $r$ such that $x_\ell < x_c < x_r$. Since the preference of agent $\ell$ is not taken into account, $\ell$ cannot affect the outcome and thus has no incentive to misreport. In addition, the mechanism always locates at $x_c$ one of the facilities that $c$ approves, and if $c$ approves both facilities, the other facility is located at $x_c+1$; hence, the cost of $c$ is always minimized. Finally, to see why agent $r$ also has no incentive to misreport, it suffices to observe that in all cases (where the preference of agent $c$ is fixed) the location of the facility that $r$ approves is either independent of her preference, or is closer to her position than if she misreports. As an example, if $c \in N_1\setminus N_2$, facility $1$ is located at $x_c$ independently from the report of $r$, and facility $2$ is located at $x_r$ if $r$ approves it. Hence, agent $r$ minimizes her cost by being truthful. The same holds for the remaining two cases.
\end{proof}

Next, we show the upper bound of $4/3$ on the approximation ratio of the mechanism.

\begin{theorem}
For instances with three agents, the $\SC$-approximation ratio of {\sc Priority-Dictatorship} is at most $4/3$.
\end{theorem}

\begin{proof}
Consider any instance $I$ with three agents $\ell$, $c$ and $r$. We distinguish between the three cases considered by the mechanism.
\begin{itemize}
\item 
If $c \in N_1 \setminus N_2$, $c$ is a median agent for facility $1$.
\begin{itemize}
\item If $r \in N_2$, the outcome is $(x_c,x_r)$, and the approximation ratio is $1$ as $r$ is a median agent for facility $2$. 
\item If $r \in N_1 \setminus N_2$, the outcome is $(x_c,x_\ell)$, and the approximation ratio is again $1$ as either $\ell$ is a median agent for facility $2$, or all agents approve only facility $1$.  
\end{itemize}

\item
If $c \in N_2 \setminus N_1$, due to symmetry to the above case, the approximation ratio is again $1$.

\item 
If $c \in N_1 \cap N_2$, $c$ is a median agent for both facilities. The approximation ratio is $1$ in the following cases: 
(a) $c$ is the {\em unique} median for both facilities, which happens when $\ell$ and $r$ approve the same set of facilities; 
(b) $r$ is a median for the facility located at $x_c+1$ (so that this facility is located in-between median agents), which happens when $\ell$ and $r$ approve a single (different) facility, or when $\ell \in N_1 \setminus N_2$ and $r \in N_1 \cap N_2$. 
So, we can consider the remaining three cases. Let $\alpha = x_c-x_\ell \geq 1$ and $\beta = x_r - x_c \geq 1$. 
\begin{itemize}
\item $\ell \in N_1 \cap N_2$, $c \in N_1 \cap N_2$, $r \in N_2 \setminus N_1$. 
One possible optimal solution is $(x_\ell,x_c)$ with social cost $2\alpha+\beta$. The solution $(x_c,x_c+1)$ computed by the mechanism has social cost $2\alpha+\beta+1$. Hence, the approximation ratio is $\frac{2\alpha+\beta+1}{2\alpha+\beta} = 1 + \frac{1}{2\alpha+\beta}$. As this is a non-increasing function in terms of $\alpha$ and $\beta$, it attains its maximum value of $4/3$ for $\alpha=\beta=1$.

\item $\ell \in N_2 \setminus N_1$, $c \in N_1 \cap N_2$, $r \in N_1 \cap N_2$.
One possible optimal solution is $(x_r,x_c)$ with social cost $\alpha+2\beta$. The solution $(x_c,x_c+1)$ computed by the mechanism has social cost $\alpha+2\beta+1$. Hence, the approximation ration is $\frac{\alpha+2\beta+1}{\alpha+2\beta} = 1 + \frac{1}{\alpha+2\beta}$. This is again a non-increasing function in terms of $\alpha$ and $\beta$, and thus attains its maximum value of $4/3$ for $\alpha=\beta=1$.

\item $\ell \in N_1 \cap N_2$, $c \in N_1 \cap N_2$, $r \in N_1 \setminus N_2$.
One possible optimal solution is $(x_c,x_\ell)$ with social cost $2\alpha+\beta$. The solution $(x_c+1,x_c)$ computed by the mechanism has social cost $2\alpha+\beta+1$. Hence, the approximation ration is $\frac{2\alpha+\beta+1}{2\alpha+\beta} = 1 + \frac{1}{2\alpha+\beta}$, which is maximized once again to $4/3$ for $\alpha=\beta=1$.
\end{itemize}
\end{itemize}
The proof is now complete.
\end{proof}


\section{Maximum cost}
We now turn our attention to the maximum cost. For this objective, \citet{serafino2016} showed an upper bound of $3$ on the approximation ratio of the TwoExtemes mechanism, and a lower bound of $3/2$ on the approximation ratio of any strategyproof mechanism. We improve both bounds, by showing a tight bound of $2$. 

\subsection{Improving the upper bound}
To achieve the improved upper bound of $2$, we consider a class of mechanisms that use only the part of the line that is occupied, from the first to the last occupied node, with possible empty nodes in-between; with some abuse of notation, we denote by $m$ the size of exactly this part of the line. These mechanisms, termed {\sc $\alpha$-Left-Right}, are parameterized by an integer $\alpha \in \{1,\ldots,m-1\}$, and their general idea is as follows: They partition the line into two parts depending on the value of $\alpha$, and then decide where to locate the facilities based on the preferences of the agents occupying nodes in these two parts. See Mechanism~\ref{mech:left-right} for a formal description. 
We first show that every {\sc $\alpha$-Left-Right} is strategyproof.

\begin{algorithm}[t]
\caption{\sc $\alpha$-Left-Right}
\label{mech:left-right}
{\bf Input:} Instance $I$ with $n$ agents\;
{\bf Output:} Feasible solution $\bz = (z_1,z_2)$\;
$L \gets$ left part of line from node $1$ to node $\alpha$\;
$N(L) \gets$ agents that occupy nodes in $L$\;
$R \gets$ right part of line from node $\alpha+1$ to node $m$\;
$N(R) \gets$ agents that occupy nodes in $R$\;
\tcp*[h]{(case 1): Each part includes agents that approve only one, different facility} \\
\uIf{$\exists \, X,Y \in \{L,R\}$: $N_1 = N(X)$ and $N_2 = N(Y)$}{
    $z_1 \gets$ median node of line defined by $N(X)$ (ties in favor of nodes farther from $\alpha$)\;
    $z_2 \gets$ median node of line defined by $N(Y)$ (ties in favor of nodes farther from $\alpha$)\; 
}
\tcp*[h]{(case 2): One part includes agents that approve only one facility} \\
\uElseIf{$\exists \, \ell \in \{1,2\}, X \in \{L,R\}$: $N_\ell \subseteq N(X)$}{
    \If{$N_\ell$ is empty}{$X \gets L$\;}
    $z_\ell \gets$ median node of line defined by $N(X)$ (ties in favor of nodes farther from $\alpha$)\;
    $z_{3-\ell} \gets \beta \in \{\alpha,\alpha+1\}\setminus X$\;
}
\tcp*[h]{(case 3): Both parts include agents from $N_1$ and $N_2$} \\
\Else{
    $z_1 \gets$ rightmost node of $L$\;
    $z_2 \gets$ leftmost node of $R$\; 
}
\end{algorithm}

\begin{theorem}
For any $\alpha \in \{1,\ldots,m-1\}$, mechanism $\alpha$-LR is strategyproof.
\end{theorem}

\begin{proof}
Consider any instance. 
We distinguish between the three cases considered by the mechanism. 

\medskip

\noindent
{\bf True preferences are as in case 1.} 
The mechanism locates facility $1$ at the median node of the line defined by $N(X)$, and facility $2$ at the median node of the line defined by $N(Y)$. 
It suffices to show that any agent $i \in N_1$ has no incentive to deviate; the case $i \in N_2$ is symmetric. If $i$ is the unique agent in $N(X)$, then she occupies the median node of the line defined by $N(X)$, where facility $1$ is located, and thus has zero cost. So, we can assume that there is some agent in $N(X) \setminus \{i\}$. If agent $i$ misreports by approving either just facility $2$ or both facilities, then we transition to case 2 with $N_1 \subseteq N(X)$, meaning that the location of facility $1$ remains the median of the line defined by $N(X)$. So, agent $i$ cannot decrease her cost, and has no incentive to deviate. 

\medskip

\noindent
{\bf True preferences are as in case 2.}
Suppose that $N_1 \subseteq N(L)$, while $N_2$ has agents in both $L$ and $R$; all other cases that fall under case 2 are symmetric. 
So, the mechanism locates facility $1$ at the median node of the line defined by $N(L)$, and facility $2$ at the leftmost node of $R$. Consider any agent $i$, and switch between all possible preferences of $i$:
\begin{itemize}
\item $i \in N_1$. 
Since all agents in $N(R)$ approve only facility $2$, we can never transition to case $3$ when $i$ misreports. Also, agent $i$ is indifferent between cases 1 and case 2 if she only approves facility $1$, and prefers case 2 to case 1 if she approves both facilities (since her position is closer to the leftmost node of $R$ than to the median node of $N(R)$). So, agent $i$ has no incentive to misreport.

\item $i \in (N_2 \setminus N_1) \cap N(L)$.
Similarly to the above case, we can never transition to case $3$ when $i$ misreports. Now, $i$ strictly prefers case 2 to case 1, as she wants facility $2$ to be located at the leftmost node of $R$, so her cost is minimized, and has no incentive to misreport.

\item $i \in (N_2 \setminus N_1) \cap N(R)$. Note that $i$ approves only facility $2$. If she misreports that she approves both facilities, then we transition to case 3, where the location of facility $2$ remains the same. If she misreports that she approves only facility $1$, then either the outcome remains the same if there is another agent in $N(R)$, or we transition to a symmetric case of case 2, where $N_2 \subseteq N(L)$ (while $N_1$ has agents in both $L$ and $R$), thus changing the location of facility $2$ from the leftmost node of $R$ to the median node of the line defined by $N(L)$. As this would increase the cost of agent $i$, she has no incentive to misreport.
\end{itemize}

\medskip

\noindent
{\bf True preferences are as in case 3.}
Since each of $N(L)$ and $N(R)$ contains agents from both $N_1$ and $N_2$, the mechanism locates facility $1$ at the rightmost node of $L$, and facility $2$ at the leftmost node of $R$.
Consider any agent $i\in N_\ell \cap N(X)$, where $\ell \in \{1,2\}$ and $X \in \{L,R\}$. 
Observe that if for each facility $j \in \{1,2\}$ there exists some agent in $N(X) \setminus \{i\}$ that approves $j$, then agent $i$ cannot affect the outcome; no matter what $i$ reports, we are still in case 3. So, we can assume that for some $j \in \{1,2\}$, all agents in $N(X)\setminus \{i\}$ approve only facility $j$. Since we are in case 3, we can also assume that $j \neq \ell$ (of course, agent $i$ might also approve $j$). To change the case considered by the mechanism, $i$ must completely agree with the other agents in $N(X)$ and report that she approves only facility $j$. This leads to a symmetric case of case 2, where $N_\ell \subseteq N(\{L,R\}\setminus X)$ (and $N_j$ contains agents in both $L$ and $R$), and hence facility $\ell$ is located at the median node of the line defined by $N(\{L,R\}\setminus X)$ and facility $j$ is still located at either $\alpha$ or $\alpha+1$. Clearly, the cost of agent $i$ can only increase as facility $\ell$ has moved farther away.
\end{proof}

Next, we focus on the approximation ratio of $\alpha$-LR mechanisms for the max cost. We distinguish between cases where the size $m$ of the line is an even or odd number, and show that there are values of $\alpha$ such that $\alpha$-LR achieves an approximation ratio of at most $2$. Before we do this, we prove a lemma providing lower bounds on the optimal max cost of a given instance, which we will use extensively. 

\begin{lemma} \label{lem:max-opt-bound}
Let $I$ be an instance. The following are true:
\begin{itemize}
\item[(a)] If there are two agents positioned at $x$ and $y > x$, and $q \in \{0,1,2\}$ is the number of facilities they both approve, then 
$$\MC^*(I) \geq q \cdot \frac{y-x}{2}.$$

\item[(b)] If there is an agent positioned at $x$ that approves both facilities, an agent positioned at $y > x$ that approves facility $1$, and an agent positioned at $z > y$ that approves facility $2$, then 
$$\MC^*(I) \geq \left\lceil \frac{y+z-2x}{3} \right\rceil.$$
\end{itemize}
\end{lemma}

\begin{proof}
To show the two properties, we use the fact that the optimal cost for an instance is at least the optimal cost when we restrict to any subset of agents and any subset of facilities, in which case we aim to balance the cost of all the agents involved. We have:

\medskip

\noindent
{\bf (a)}
We begin with the case $q=1$ as the claim holds trivially when $q=0$. Clearly, if we place the facility before $x$ or after $y$, the claim holds. So, let us assume that we place the facility at node $a$ such that $x\leq a \leq y$. The cost of the agent at node $x$ is then (at least) $a-x$, while the cost of the agent at node $y$ is (at least) $y-a$. The claim follows since it cannot be that both $a-x$ and $y-a$ are strictly less than $\frac{y-x}{2}$.

When $q=2$, the claim follows if at least one facility is placed before $x$ or after $y$. Let us assume that we place the facilities at nodes $a$ and $b$ such that $x\leq \min\{a,b\}< \max\{a,b\}\leq y$. The cost of the agent at node $x$ is then $a+b-2x$, while the cost of the agent at node $y$ is $2y-a-b$. The claim follows since it cannot be that both $a+b-2x$ and $2y-a-b$ are strictly less than $y-x$.

\medskip

\noindent
{\bf (b)}
In this case, we want to locate facility $1$ at some node $a \in [x,y]$ and facility $2$ at some node $b \in [x,z]$, such that the maximum cost among the three agents is minimized. 
The cost of the agent at node $x$ is then $a+b-2x$, the cost of the agent at $y$ is $y-a$, while the cost of the agent at $z$ is $z-b$. As the sum of costs equals $y+z-2x$, it cannot be the case that all three costs are strictly less than $\frac{y+z-2x}{3}$. The claim follows since any cost must be an integer.
\end{proof}

We are now ready to bound the approximation ratio of particular $\alpha$-LR mechanisms. We start with instances where $m$ is an even number, for which we use $\alpha=m/2$; that is, we partition the line into two parts of equal size. 

\begin{theorem}
When $m$ is even, the $\MC$-approximation ratio of $m/2$-LR is at most $2$.
\end{theorem}

\begin{proof}
Consider any instance $I$.
We distinguish between the three cases considered by the mechanism.

\medskip

\noindent
{\bf Case 1.}
Since the agents in $N(X)$ approve only facility $1$ and the agents in $N(Y)$ approve only facility $2$, locating facility $1$ at the median node of the line defined by $N(X)$, and facility $2$ at the median node of the line defined by $N(Y)$ is the optimal solution. 

\medskip

\noindent
{\bf Case 2.}
Suppose that $N_1 \subseteq N(L)$ and that $N_2$ contains agents in both $N(L)$ and $N(R)$; this is one of the symmetric instances captured by case 2. The mechanism locates facility $1$ at the median node $y_L$ (with $1 \leq y_L \leq \lfloor \frac{m+2}{4}\rfloor$) of the line defined by $N(L)$, and facility $2$ at node $\frac{m}{2}+1$ (the leftmost node of $R$). We distinguish between the following cases depending on the preferences of the agents with the maximum cost for the solution $\bz$ computed by the mechanism.
\begin{itemize}
\item {\bf The cost of the mechanism is equal to the cost of an agent that approves a single facility.}  
As all agents that approve facility $1$ are in $N(L)$, and facility $1$ is located at the median of the line defined by $N(L)$, the cost of any agent that approves only facility $1$ can be at most $\max\{\lfloor\frac{m+2}{4}\rfloor-1, \frac{m}{2}-\lfloor\frac{m+2}{4}\rfloor\} \leq \frac{m}{4}$. 
Since facility $2$ is located at node $\frac{m}{2}+1$, the cost of any agent that approves only facility $2$ can be at most $\frac{m}{2}+1-1=\frac{m}{2}$. 
Hence, $\MC(\bz|I) \leq \frac{m}{2}$. As $N_2$ contains at least one agent in $N(L)$, there exists at least one agent at a node $x \leq \frac{m}{2}$ that approves facility $2$. By applying Lemma~\ref{lem:max-opt-bound}(a) with $x$ and $y=m$, we have that $\MC^*(I) \geq \frac{m-x}{2} \geq \frac{m}{4}$, yielding that the approximation ratio is at most $2$. 

\item {\bf The cost of the mechanism is equal to the cost of an agent that approves both facilities.}  
Since we are in case 2 with $N_1 \subseteq N_L$, let $x \leq m/2$ be the position of the agent $i$ that approves both facilities and has the maximum cost among all such agents. The cost of agent $i$, and thus of the mechanism, is
$\MC(\bz|I) = |x-y_L|+\frac{m}{2}+1-x$.

If $x > y_L$, we have $\MC(\bz|I) = \frac{m}{2}+1-y_L \leq \frac{m}{2}$. As in the case where the cost is due to an agent that approves a single facility, we have $\MC^*(I) \geq \frac{m}{4}$, and thus the approximation ratio is at most $2$.

Otherwise, if $x \leq y_L$, we have $\MC(\bz|I) = \frac{m}{2}+1+y_L-2x \leq \frac{3(m+2)}{4}-2x$. Since agent $i$ and the agent at node $m$ both approve facility $2$, by Lemma~\ref{lem:max-opt-bound}(a), we have that $\MC^*(I) \geq \frac{m-x}{2}$. Hence, the approximation ratio is at most $\frac{3m+6-8x}{2m-2x}$. As this is a non-increasing function in terms of $x$, it attains its maximum value of $\frac{3m-2}{2m-2}$ for $x=1$. For every $m \geq 2$, it holds that $\frac{3m-2}{2m-2} \leq 2$.
\end{itemize}

\medskip

\noindent
{\bf Case 3.}
Recall that the mechanism locates facility $1$ at $\frac{m}{2}$ (rightmost node of $L$), and facility $2$ at $\frac{m}{2}+1$ (leftmost node of $R$). 
Without loss of generality, we can assume that the agent at node $1$ approves facility $1$ (and possibly also facility $2$). 
We switch between the following two subcases:
\begin{itemize}
\item {\bf The cost of the mechanism is equal to the cost of an agent that approves a single facility.} 
Then, $\MC(\bz|I) \leq \frac{m}{2}$ (the distance between node $1$ and node $\frac{m}{2}+1$). As we are in case $3$, there exists an agent at some node $y \geq \frac{m}{2}+1$ that approves facility $1$, and by our assumption that the agent at node $1$ approves facility $1$, Lemma~\ref{lem:max-opt-bound}(a) gives $\MC^*(I) \geq \frac{y-1}{2} \geq \frac{m}{4}$. So, the approximation ratio is at most $2$. 

\item {\bf The cost of the mechanism is equal to the cost of an agent that approves both facilities.} 
Without loss of generality, let $x \leq m/2$ be the position of the agent $i$ that has the maximum cost among all agents that approve both facilities. 
Then, $\MC(\bz|I) = \frac{m}{2}-x + \frac{m}{2}+1 - x = m+1-2x$. As the agent at node $m$ approves some facility that is also approved by $i$, by Lemma~\ref{lem:max-opt-bound}(a), we get $\MC^*(I) \geq \frac{m-x}{2}$. The approximation ratio is $2\cdot \frac{m+1-2x}{m-x}$, which is a non-increasing function in terms of $x$, and attains its maximum value of $2$ for $x=1$.
\end{itemize}

In any case, the approximation ratio of the mechanism is $2$, and the theorem follows.
\end{proof}

For instances with odd $m$, we use $\alpha=(m+1)/2$. The proof of the following theorem is similar in structure with the previous theorem for even $m$, but is slightly more complicated. 

\begin{theorem}
When $m$ is odd, the $\MC$-approximation ratio of $(m+1)/2$-LR is at most $2$.
\end{theorem}

\begin{proof}
Consider any instance $I$.
We distinguish between the three cases considered by the mechanism.

\medskip

\noindent
{\bf Case 1.}
Since the agents in $N(X)$ approve only facility $1$, and the agents in $N(Y)$ approve only facility $2$, locating facility $1$ at the median node of the line defined by $N(X)$, and facility $2$ at the median node of the line defined by $N(Y)$ is the optimal outcome. 

\medskip

\noindent
{\bf Case 2.}
Suppose that $N_1 \subseteq N(L)$ and that $N_2$ contains agents in both $N(L)$ and $N(R)$; this is one of the symmetric instances captured by case 2. The mechanism locates facility $1$ at the median node $y_L$ (with $1 \leq y_L \leq \lfloor \frac{m+3}{4}\rfloor$) of the line defined by $N(L)$, and facility $2$ at node $\frac{m+1}{2}+1 =\frac{m+3}{2}$ (the leftmost node of $R$). We distinguish between the following cases depending on the preferences of the agents with the maximum cost for the solution $\bz$ computed by the mechanism.

\begin{itemize}
\item {\bf The cost of the mechanism is equal to the cost of an agent that approves a single facility.}  
As all agents that approve facility $1$ are in $N(L)$, and facility $1$ is located at the median of the line defined by $N(L)$, the cost of any agent that approves only facility $1$ can be at most $\max\{\lfloor \frac{m+3}{4}\rfloor-1,\frac{m+1}{2} - \lfloor \frac{m+3}{4}\rfloor\} \leq \frac{m+1}{4}$. Let $x \leq \frac{m+1}{2}$ be the position of the leftmost agent that approves facility $2$. Since the agent at node $m$ also approves facility $2$, by Lemma~\ref{lem:max-opt-bound}(a), we have that $\MC^*(I) \geq \frac{m-x}{2}$. We now distinguish between two subcases, based on the value of $x$. 

If $x = \frac{m+1}{2}$, the maximum cost among agents that approve facility $2$ is at most $m - \frac{m+3}{2} = \frac{m-3}{2}$, and thus $\MC(\bz|I) \leq \max\left\{ \frac{m+1}{4}, \frac{m-3}{2} \right\}$. Since $\MC^*(I) \geq \frac{m-1}{4}$, the approximation ratio is at most $2$. 

Otherwise, if $x \leq \frac{m-1}{2}$, the maximum cost among agents that approve facility $2$ is a most $\frac{m+3}{2}-1=\frac{m+1}{2}$, and thus $\MC(\bz|I) \leq \max\left\{ \frac{m+1}{4}, \frac{m+1}{2} \right\}=\frac{m+1}{2}$. Since $\MC^*(I) \geq \frac{m-x}{2}\geq \frac{m+1}{4}$, the approximation ratio is again at most $2$.

\item {\bf The cost of the mechanism is equal to the cost of an agent that approves both facilities.}  
Since we are in case 2 with $N_1 \subseteq N_L$, let $x \leq \frac{m+1}{2}$ be the position of the agent $i$ that approves both facilities and has the maximum cost among all such agents. The cost of agent $i$, and thus of the mechanism, is 
$\MC(\bz|I) = |x-y_L|+\frac{m+3}{2}-x$.

If $x \leq y_L$, then since $y_L \leq \frac{m+3}{4}$, we have that $\MC(\bz|I) = \frac{m+3}{2}+y_L-2x \leq \frac{3(m+3)}{4}-2x$. As node $m$ is occupied by an agent that approves facility $2$, by Lemma~\ref{lem:max-opt-bound}(a), we have that $\MC^*(I) \geq \frac{m-x}{2}$, and thus the approximation ratio is $\frac{3m+9-8x}{2m-2x}$. This is a non-increasing function of $x \geq 1$, and attains its maximum value of $\frac{3m+1}{2m-2}$ for $x=1$. For every $m \geq 5$, it holds that $\frac{3m+1}{2m-2} \leq 2$. When $m=3$, for $x \leq y_L$ to be true, it has to be the case that $x=y_L=1$; so, the cost of agent $i$ for $\bz$ is $2$, while $\MC^*(I) \geq 1$, leading to an approximation ratio of at most $2$.

Otherwise, if $x > y_L$, for $y_L = 1$ to be possible, it would have to be the case that $m=3$ and $x=2$; then, the cost of agent $i$ is $2$, while $\MC^*(I) \geq 1$, and so the approximation ratio is at most $2$. Hence, assume that $y_L \geq 2$. Then, we have $\MC(\bz|I)=\frac{m+3}{2}-y_L \leq \frac{m-1}{2}$. Since $x \leq \frac{m+1}{2}$ and the agent at node $m$ approves facility $2$, by Lemma~\ref{lem:max-opt-bound}(a), we have that $\MC^*(I) \geq \frac{m-1}{4}$, and the approximation ratio is at most $2$. 
\end{itemize}

\medskip

\noindent
{\bf Case 3.}
Recall that in this case the mechanism locates facility $1$ at $\frac{m+1}{2}$, and facility $2$ at $\frac{m+3}{2}$. 
Without loss of generality, we assume that the agent at node $1$ approves facility $1$ (and possibly also facility $2$). 
We switch between two subcases:
\begin{itemize}
\item {\bf The cost of the mechanism is equal to the cost of an agent that approves a single facility.} 
Then, $MC(\bz|I) \leq \frac{m+3}{2}-1 = \frac{m+1}{2}$. As we are in case $3$, there exists an agent at some node $y \geq \frac{m+3}{2}$ that approves facility $1$, and by our assumption that the agent at node $1$ approves facility $1$, Lemma~\ref{lem:max-opt-bound}(a) gives $\MC^*(I) \geq \frac{y-1}{2} \geq \frac{m+1}{4}$, and the approximation ratio is at most $2$. 

\item {\bf The cost of the mechanism is equal to the cost of an agent that approves both facilities.} 
Without loss of generality, let $x \leq \frac{m+1}{2}$ be the position of the agent $i$ that has the maximum cost among agents that approve both facilities. Then, $\MC(\bz|I) = \frac{m+1}{2}-x + \frac{m+3}{2} - x = m-2x+2$. Since we are in case $3$, in $N(R)$, there exists an agent that approves facility $1$ and an agent that approves facility $2$. Consider the following two subcases:

If there is an agent $j$ at some node $y \geq \frac{m+3}{2}$ that approves both facilities, then, by Lemma~\ref{lem:max-opt-bound}(a), $\MC^*(I) \geq y-x \geq \frac{m+3}{2}-x$, and the approximation ratio is at most $2 \cdot \frac{m-2x+2}{m-2x+3} \leq 2$.

Otherwise, if there is no agent in $N(R)$ that approves both facilities, suppose that the agent at node $m$ approves facility $2$, and there exists an agent at some node $y \in \left[ \frac{m+3}{2}, m \right)$ that approves facility $1$. If $x=1$, then $\MC(\bz|I) \leq m$ and, by Lemma~\ref{lem:max-opt-bound}(b), $\MC^*(I) \geq \lceil \frac{y+m-2}{3} \rceil \geq \lceil \frac{m}{2} - \frac{1}{6} \rceil = \frac{m+1}{2}$; hence, the approximation ratio is at most $2$. If $x \geq 2$, it suffices to use the bound $\MC^*(I) \geq \frac{m-x}{2}$ implied by Lemma~\ref{lem:max-opt-bound}(a), to get an upper bound of $2\cdot \frac{m-2x+2}{m-x}$ on the approximation ratio. This is a non-increasing function of $x$, and thus attains its maximum value of $2$ for $x = 2$.
\end{itemize}

In any case, the approximation ratio is at most $2$, and the proof is complete.
\end{proof}

\subsection{A tight lower bound for deterministic mechanisms}

We conclude the presentation of our technical results with a tight lower bound of $2$ on the approximation ratio of any strategyproof mechanism with respect to the maximum cost objective.

\begin{theorem} \label{thm:max-lower}
The $\MC$-approximation ratio of any strategyproof mechanism is at least $2$.
\end{theorem}

\begin{proof}
Suppose that there exists a strategyproof mechanism $M$ with approximation ratio strictly smaller than $2$. We will reach a contradiction by examining a series of instances, all of which involve three agents and no empty nodes; see also Figure \ref{fig:MC-lb}.

We begin with instance $I_1$, in which the first and third agents approve only facility $1$, while the second agent approves only facility $2$. 
Clearly, $M$ must return either $(2,3)$ or $(2,1)$ as $\MC((2,3)|I_1) = \MC((2,1)|I_1) = 1$; any solution where facility $1$ is not placed at the second node has maximum cost $2$, and returning such a solution would contradict the assumption that the approximation ratio of $M$ is strictly smaller than $2$. 
Without loss of generality, let us assume that $M$ returns the solution $(2,3)$.

Next, consider instance $I_2$, in which the first agent approves only facility $1$, while the remaining agents approve only facility $2$. $M$ must output either $(2,3)$ or $(1,3)$ due to strategyproofness. Indeed, any solution where facility $2$ is not placed at the third node leads to a cost of at least $1$ for the third agent. But then, that agent would misreport that she only approves facility $1$, thus leading to instance $I_1$, and obtain a cost of $0$ for the resulting solution $(2,3)$.

If $M$ returns $(2,3)$ for instance $I_2$, consider instance $I_3$, in which the first agent approves both facilities, while the other two agents approve facility $2$. $M$ must return the optimal solution $(1,2)$ with $\MC((1,2)|I_3) = 1$, since any other solution leads to a maximum cost of at least $2$. In this case, however, the first agent in $I_2$ would misreport that she approves both facilities to reduce her cost from $1$ to $0$; this contradicts the assumption that $M$ is strategyproof.

Otherwise, when $M$ returns $(1,3)$ for $I_2$, consider instance $I_4$, in which the first agent approves only facility $1$, the second agent approves both facilities, and the last agent approves only facility $2$. There are two optimal solutions in $I_4$, $(2,3)$ and $(1,2)$, with $\MC((2,3)|I_4) = \MC((1,2)|I_4) = 1$; any other solution has maximum cost $2$. Out of these solutions, $(1,2)$ would give the second agent in $I_2$ incentive to misreport that she approves both facilities to reduce her cost from $1$ to $0$. Hence, $M$ must return $(2,3)$ when given as input $I_4$. To conclude the proof, consider  instance $I_5$, in which the first two agents approve both facilities, while the third agent approves only facility $2$. The optimal solution is $(1,2)$ with $\MC((1,2)|I_5)=1$; any other solution has maximum cost of at least $2$. But then, the first agent in $I_4$ has incentive to misreport that she approves both facilities to reduce her cost from $1$ to $0$; this again contradicts the fact that $M$ is strategyproof.
\end{proof}

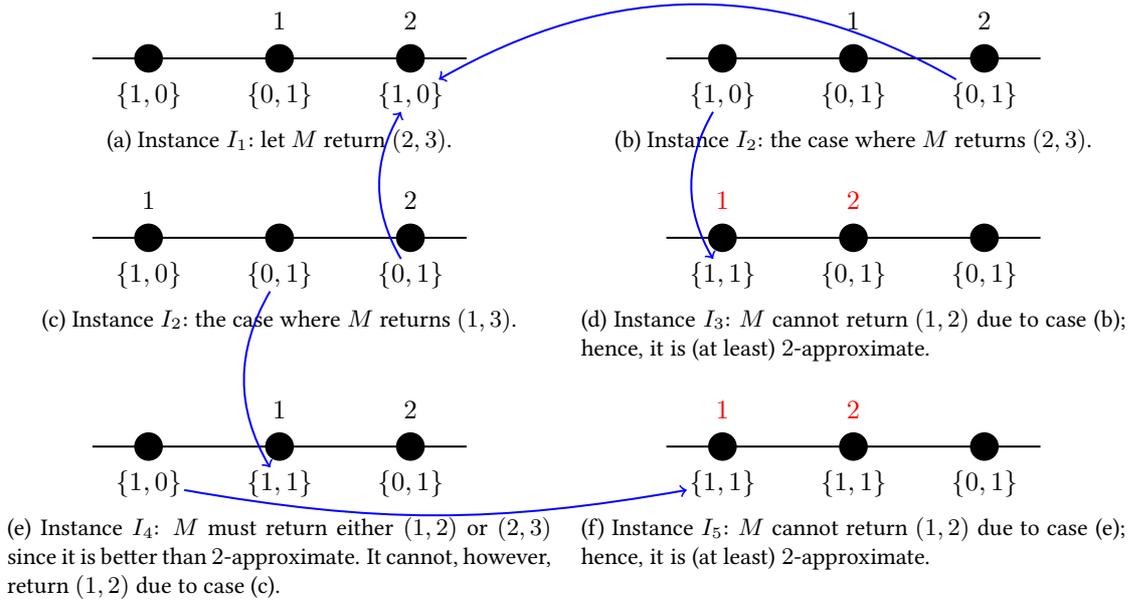
\begin{figure}[t]
\tikzset{every picture/.style={line width=0.75pt}} 
\begin{subfigure}[t]{0.45\linewidth}
\centering
\begin{tikzpicture}[x=0.7pt,y=0.7pt,yscale=-1,xscale=1, remember picture]
\draw [line width=0.75]  (0,0) -- (200,0) ;
\filldraw (30,0) circle (5pt);
\filldraw (100,0) circle (5pt);
\filldraw (170,0) circle (5pt);

\draw (30,20) node [inner sep=0.75pt] [font=\small]  {$\{1,0\}$};

\draw (100,-20) node [inner sep=0.75pt] [font=\small]  {$1$};
\draw (100,20) node [inner sep=0.75pt] [font=\small]  {$\{0,1\}$};

\draw (170,-20) node [inner sep=0.75pt] [font=\small]  {$2$};
\draw (170,20) node [inner sep=0.75pt]  (blue-a3) [font=\small]  {$\{1,0\}$};
\end{tikzpicture}
\caption{Instance $I_1$: let $M$ return $(2,3)$.}
\end{subfigure}
\ \ \ 
\begin{subfigure}[t]{0.45\linewidth}
\centering
\begin{tikzpicture}[x=0.7pt,y=0.7pt,yscale=-1,xscale=1, remember picture]
\draw [line width=0.75]  (0,0) -- (200,0) ;
\filldraw (30,0) circle (5pt);
\filldraw (100,0) circle (5pt);
\filldraw (170,0) circle (5pt);

\draw (30,20) node [inner sep=0.75pt]  (blue-b1) [font=\small]  {$\{1,0\}$};

\draw (100,-20) node [inner sep=0.75pt]  [font=\small]  {$1$};
\draw (100,20) node [inner sep=0.75pt]  [font=\small]  {$\{0,1\}$};

\draw (170,-20) node [inner sep=0.75pt]   [font=\small]  {$2$};
\draw (170,20) node [inner sep=0.75pt]  (blue-b3) [font=\small]  {$\{0,1\}$};
\end{tikzpicture}
\caption{Instance $I_2$: the case where $M$ returns $(2,3)$.}
\end{subfigure}

\bigskip

\begin{subfigure}[t]{0.45\linewidth}
\centering
\begin{tikzpicture}[x=0.7pt,y=0.7pt,yscale=-1,xscale=1, remember picture]
\draw [line width=0.75]  (0,0) -- (200,0) ;
\filldraw (30,0) circle (5pt);
\filldraw (100,0) circle (5pt);
\filldraw (170,0) circle (5pt);

\draw (30,-20) node [inner sep=0.75pt]  [font=\small]  {$1$};
\draw (30,20) node [inner sep=0.75pt]  [font=\small]  {$\{1,0\}$};

\draw (100,20) node [inner sep=0.75pt] (blue-c2) [font=\small]  {$\{0,1\}$};

\draw (170,-20) node [inner sep=0.75pt]  [font=\small]  {$2$};
\draw (170,20) node [inner sep=0.75pt]  (blue-c3) [font=\small]  {$\{0,1\}$};
\end{tikzpicture}
\caption{Instance $I_2$: the case where $M$ returns $(1,3)$.}
\end{subfigure}
\ \ \ 
\begin{subfigure}[t]{0.45\linewidth}
\centering
\begin{tikzpicture}[x=0.7pt,y=0.7pt,yscale=-1,xscale=1, remember picture]
\draw [line width=0.75]  (0,0) -- (200,0) ;
\filldraw (30,0) circle (5pt);
\filldraw (100,0) circle (5pt);
\filldraw (170,0) circle (5pt);

\draw (30,-20) node [inner sep=0.75pt]  [font=\small]  {\textcolor{red}{$1$}};
\draw (30,20) node [inner sep=0.75pt]  (blue-d1) [font=\small]  {$\{1,1\}$};

\draw (100,-20) node [inner sep=0.75pt]  [font=\small]  {\textcolor{red}{$2$}};
\draw (100,20) node [inner sep=0.75pt] (blue-d2) [font=\small]  {$\{0,1\}$};

\draw (170,20) node [inner sep=0.75pt]  [font=\small]  {$\{0,1\}$};
\end{tikzpicture}
\caption{Instance $I_3$: $M$ cannot return $(1,2)$ due to case (b); hence, it is (at least) $2$-approximate.}
\end{subfigure}

\bigskip

\begin{subfigure}[t]{0.45\linewidth}
\centering
\begin{tikzpicture}[x=0.7pt,y=0.7pt,yscale=-1,xscale=1, remember picture]
\draw [line width=0.75]  (0,0) -- (200,0) ;
\filldraw (30,0) circle (5pt);
\filldraw (100,0) circle (5pt);
\filldraw (170,0) circle (5pt);

\draw (30,20) node [inner sep=0.75pt] (blue-e1) [font=\small]  {$\{1,0\}$};

\draw (100,-20) node [inner sep=0.75pt]  [font=\small]  {$1$};
\draw (100,20) node [inner sep=0.75pt]  (blue-e2) [font=\small]  {$\{1,1\}$};

\draw (170,-20) node [inner sep=0.75pt]  [font=\small]  {$2$};
\draw (170,20) node [inner sep=0.75pt]  [font=\small]  {$\{0,1\}$};
\end{tikzpicture}
\caption{Instance $I_4$: $M$ must return either $(1,2)$ or $(2,3)$ since it is better than $2$-approximate. It cannot, however, return $(1,2)$ due to case (c).}
\end{subfigure}
\ \ \ 
\begin{subfigure}[t]{0.45\linewidth}
\centering
\begin{tikzpicture}[x=0.7pt,y=0.7pt,yscale=-1,xscale=1, remember picture]
\draw [line width=0.75]  (0,0) -- (200,0) ;
\filldraw (30,0) circle (5pt);
\filldraw (100,0) circle (5pt);
\filldraw (170,0) circle (5pt);

\draw (30,-20) node [inner sep=0.75pt] [font=\small]  {\textcolor{red}{$1$}};
\draw (30,20) node [inner sep=0.75pt]  (blue-f1) [font=\small]  {$\{1,1\}$};

\draw (100,-20) node [inner sep=0.75pt]  [font=\small]  {\textcolor{red}{$2$}};
\draw (100,20) node [inner sep=0.75pt]  (blue-f2) [font=\small]  {$\{1,1\}$};

\draw (170,20) node [inner sep=0.75pt]  [font=\small]  {$\{0,1\}$};
\end{tikzpicture}
\caption{Instance $I_5$: $M$ cannot return $(1,2)$ due to case (e); hence, it is (at least) $2$-approximate.}
\end{subfigure}
\\\ 
\caption{The instances used in the proof of Theorem \ref{thm:max-lower}. Each instance has $3$ agents and no empty nodes. The agent preferences appear below each node, while the facility assignment appears above the nodes. A black font denotes the mechanism's assignment, while a red font denotes an optimal but excluded assignment. Blue arrows denote how instances are related when a single agent's preferences change.}
\label{fig:MC-lb}
    \begin{tikzpicture}[remember picture,overlay]
        \path[->,blue] (blue-b3) edge [bend right=30] (blue-a3);
        \path[->,blue] (blue-c3) edge [bend left=30] (blue-a3);
        \path[->,blue] (blue-b1) edge [bend right=30] (blue-d1);
        \path[->,blue] (blue-c2) edge [bend right=30] (blue-e2);
        \path[->,blue] (blue-e1) edge [bend right=10] (blue-f1);
    \end{tikzpicture}
\end{figure}


\section{Conclusion and open problems}
In this paper, we revisited the discrete truthful heterogeneous two-facility location problem, and showed bounds on the approximation ratio of deterministic strategyproof mechanisms that greatly improve upon the previous best-known ones, both with respect to the social cost as well as the maximum cost. There are still many open questions and directions for future research. 

While we were able to show a bound for the social cost that is a small constant, thus improving upon the linear bound that was previously known, we were unable to close the gap between the lower bound of $4/3$ and the upper bound of $17/4$. Besides deterministic mechanisms, it would also be interesting to focus on randomized mechanisms and the maximum cost objective, for which there is a gap between the lower bound of $4/3$ and the upper bound of $3/2$ shown by \citet{serafino2016}. 

Going beyond the particular model studied here, there are many extensions to be considered. One could study settings with more than just two facilities, settings where the positions of the agents are their private information and can report empty nodes as their positions, settings with different heterogeneous preferences such as fractional or obnoxious ones, and also settings with more general location graphs such as trees or regular graphs.

\bibliographystyle{plainnat}
\bibliography{references}
\end{document}